%% file: main.tex
\newif\ifpublic
\publicfalse 

\documentclass[11pt]{article}
\usepackage[pagebackref,colorlinks=true,pdfpagemode=none,linkcolor=blue,,urlcolor=blue,citecolor=blue,pdfstartview=FitH]{hyperref}

\ifpublic
\usepackage[disable]{todonotes}
\else
\usepackage[colorinlistoftodos]{todonotes}
\fi

\newcommand{\swastik}[1]{\todo[color=red!20!blue!15,inline,size=\footnotesize]{Swastik: #1}}
\newcommand{\sriprahladhuvacha}[1]{\todo[color=red!100!green!33,inline,size=\small]{ph: #1}}

\usepackage[utf8]{inputenc}
\usepackage[russian,english]{babel}
\usepackage{amsmath}
\usepackage{amssymb}
\usepackage{amsthm}
\usepackage{bbm}
\usepackage{fullpage}
\usepackage{mathpazo}
\usepackage{blkarray}
\usepackage{bm}
\usepackage[capitalize,nameinlink]{cleveref}
\renewcommand{\eqref}[1]{\hyperref[#1]{(\ref*{#1})}}

\theoremstyle{theorem}
\newtheorem{theorem}{Theorem}[section]
\newtheorem{corollary}[theorem]{Corollary}
\newtheorem{lemma}[theorem]{Lemma}

\newtheorem{definition}[theorem]{Definition}
\newtheorem{claim}[theorem]{Claim}
\newtheorem{fact}[theorem]{Fact}
\theoremstyle{definition}
\newtheorem{remark}[theorem]{Remark}

\definecolor{darkred}{rgb}{0.5,0,0}
\definecolor{darkgreen}{rgb}{0,0.5,0}

\newcommand{\F}{{\mathbb F}}
\DeclareMathOperator{\E}{{\mathbb E}}
\newcommand{\ip}[1]{\left\langle #1 \right\rangle}
\newcommand{\abs}[1]{\ensuremath{\left\lvert #1 \right\rvert}}

\newcommand{\eps}{\varepsilon}
\renewcommand{\epsilon}{\eps}
\newcommand{\set}[1]{\left\{{#1}\right\}}

\newcommand{\vspan}{\mathrm{span}}
\newcommand{\bias}{\mathrm{bias}}
\newcommand{\poly}{\mathrm{Poly}}
\newcommand{\corr}{\mathrm{Corr}}
\newcommand{\smpoly}{\mathrm{smPoly}}

\newcommand{\cP}{\mathcal{P}}

\newcommand{\rank}{\mathrm{rank}}

\newcommand{\tr}{\mathrm{Trace}}
\newcommand{\sqbinom}{\genfrac{[}{]}{0pt}{}}
\newcommand{\calP}{\mathcal{P}}

\title{
On Multilinear Forms: Bias, Correlation, and Tensor Rank
}

\author{%
  Abhishek Bhrushundi\thanks{Dept. of Computer Science, Rutgers University, U.S.A. \texttt{abhishek.bhr@gmail.com}.} \and 
  Prahladh Harsha\thanks{Tata Institute of Fundamental Research,
    India.  \texttt{prahladh@tifr.res.in}. This work was done when the
  author was visiting Rutgers University/DIMACS, USA and Weizmann
  Institute of Science, Israel. This work was partially supported by
  the DIMACS/Simons Collaboration in Cryptography through NSF grant
  \#CNS-1523467 and the Israel-India ISF-UGC grant.} \and
  Pooya Hatami\thanks{Dept. of Computer Science, University of Texas at Austin, U.S.A.  \texttt{pooyahat@gmail.com}. Part of this work was done when the author was a postdoc at DIMACS. Supported by a Simons Investigator Award (\#409864, David Zuckerman)} \and 
  Swastik Kopparty\thanks{Dept. of Computer Science \& Dept. of Mathematics, Rutgers University, U.S.A. Research supported in part by NSF grants CCF-1253886 and CCF-1540634. \texttt{swastik.kopparty@gmail.com}.} \and 
  Mrinal Kumar\thanks{Center for Mathematical Sciences and Applications, Harvard University,  U.S.A.  \texttt{mrinalkumar08@gmail.com}.}
}

\date{}

\begin{document}

\sloppy

\maketitle

\setcounter{page}{1}
\thispagestyle{empty}

\begin{abstract}
In this paper, we prove new relations between the bias of multilinear forms, the correlation between multilinear forms and lower degree polynomials,
and the rank of tensors over $\F_2= \{0,1\}$. We show the following results for multilinear
forms and tensors.
\paragraph*{Correlation bounds. } We show that a random $d$-linear form has  exponentially low correlation with low-degree polynomials. More
  precisely, for $d \ll 2^{o(k)}$, we show that a random $d$-linear form
  $f(X_1,X_2, \dots, X_d) : \left(\F_2^{k}\right)^d \rightarrow \F_2$ has correlation
  $2^{-k(1-o(1))}$ with any polynomial of degree at most $d/2$. 

This result is proved by giving near-optimal bounds on the bias of
a random $d$-linear form, which is in turn proved by giving near-optimal
bounds on the probability that a random rank-$t$ $d$-linear form is
identically zero. 
\paragraph*{Tensor-rank vs Bias. } We show that if a $d$-dimensional tensor has small rank,
  then the bias of the associated $d$-linear form is large. More precisely,
given any $d$-dimensional
  tensor $$T :\underbrace{[k]\times \ldots [k]}_{\text{$d$ times}}\to
  \F_2$$ of rank at most $t$, the bias of the associated $d$-linear form
  $$f_T(X_1,\ldots,X_d) := \sum_{(i_1,\dots,i_d) \in [k]^d}
  T(i_1,i_2,\ldots, i_d) X_{1,i_1}\cdot X_{1,i_2}\cdots X_{d,i_d}$$ is
  at least  $\left(1-\frac1{2^{d-1}}\right)^t$. 

The above bias vs tensor-rank connection suggests a natural approach to
proving nontrivial tensor-rank lower bounds for $d=3$. In particular,
we use this approach to prove that the finite field multiplication tensor has
tensor rank at least $3.52 k$ matching the best known lower bound for any explicit tensor in three dimensions over $\F_2$.
\end{abstract}



\section{Introduction}

This work is motivated by two fundamental questions regarding ``explicit constructions" in complexity theory: finding functions
uncorrelated with low degree polynomials, and finding tensors with high tensor rank.
\paragraph{Functions uncorrelated with low degree polynomials. }

The first question is that of finding an explicit function uncorrelated with low degree polynomials. More concretely, we seek functions $f: \F_2^n \to \F_2$ 
such that for every polynomial $P(X_1, \ldots, X_n) \in \F_2[X_1, \ldots, X_n]$ of degree at most $\ell$ (assume $\ell \approx n^{0.1}$ say), 
$$ \Pr_{x \in \F_2^n} [ f(x) = P(x) ] \leq \frac{1}{2} + \epsilon_n.$$
It is well known (and easy to prove) that a random function $f$ has this property with $\epsilon_n$ superpolynomially small (and even exponentially small);
the challenge is to find an explicit function $f$.

A solution to this problem will have immediate applications in Boolean circuit complexity. It will give hard-on-average problems for $AC^0(\oplus)$, and via the
Nisan-Wigderson hardness vs. randomness technique~\cite{NisanW1994}, it will give pseudorandom generators against $AC^0(\oplus)$ (improving upon analogous results for $AC^0$ from the
late 1980s). The original motivation for an explicit function with
small $\eps_n$ came from the seminal work of
Razborov~\cite{Razborov1987} and Smolensky~\cite{Smolensky1987} who
showed that any function computable by a sub-exponential sized
$AC^{0}(\oplus)$ circuit satisfies $\epsilon_n = \Omega(1)$ and
  furthermore that the $MOD_3$ has $\epsilon_n = O(1)$. The
  Nisan-Wigderson paradigm~\cite{NisanW1994} of pseudorandom
  generator construction requires explicit functions with exponentially small
  $\epsilon_n$. The current best known constructions of explicit
  functions~\cite{Razborov1987,Smolensky1987,BenSassonK2012,ViolaW2008} that cannot be approximated by low-degree polynomials come
  in two flavors, (a) polynomially small $\epsilon_n$ (in fact,
  $O(1/\sqrt{n})$) for large degree bounds ($d$ as large as $n^{0.1}$)
  or (b) exponentially small $\epsilon_n$ for small degree bounds ($d<<\log n$). However, we do not know of any explicit function $f$
  that exhibits exponentially small $\epsilon_n$ against
  low-degree polynomials of polynomially large (or even
  super-logarithmically large) degree polynomials. For a nice survey
  on correlation with low degree polynomials, see \cite{Viola2009}.
\paragraph{Tensors with high rank. }
The second question is that of finding an explicit tensor of high
tensor rank.  Tensors are a high-dimensional generalization of
($2$-dimensional) matrices. Just as a matrix of size $k$ over a field
$\F$ is given by a map $M:[k]^2 \to \F$, a tensor $T$ of dimension $d$
and size $k$ is given by a map $T:[k]^d \to \F$. A tensor $T$ is said
to be of rank one if there exist vectors
$u_1, u_2, \ldots, u_d \in \F_2^k$ such that
$T= u_1 \otimes u_2 \otimes \cdots \otimes u_d$ or equivalently, for
all $(i_1,\dots,i_d) \in [k]^d$, we have
$T(i_1,\dots,i_d) = u_{1,i_1}\cdot u_{2,i_2}\cdots u_{d,i_d}$. A
tensor $T$ is said to be of tensor-rank at most $t$ if it can be written as
the sum of $t$ rank one tensors. We seek
tensors with tensor-rank as high as possible.

It is well known (and easy to prove) that a random tensor $T$ has tensor rank $t$ as large as $\Omega(k^{d-1}/d)$. The challenge is to
find an explicit such $T$ with tensor rank larger than $k^{\lfloor\frac{d}{2}\rfloor}$.  A substantial improvement on this lower bound for any explicit tensor will have
immediate applications in arithmetic circuit complexity; for $d = 3$, 
it will give improved arithmetic circuit lower
bounds~\cite{Strassen1973b}, and for large $d$ it will give
superpolynomial arithmetic formula lower
bounds~\cite{Raz2013,ChillaraKSV2016}. For general \emph{odd} $d$, a lower bound of $2k^{\lfloor d/2 \rfloor} + k - O(d\log k)$ was shown  for an explicit tensor by Alexeev et al.~\cite{AFT11}, while for \emph{even} $d$, no lower bounds better than the trivial bound $k^{\lfloor\frac{d}{2}\rfloor}$ are known for any explicit tensor.

 Unlike matrix rank, we do not
have a good understanding of tensor-rank even for 3-dimensional
tensors. For instance, it is known that for a given 3-dimensional
tensor $T$ over the rationals, the problem of deciding if the rank of
$T$ is at most $k$ is  NP-hard~\cite{Hastad1990}.  In the
case of dimension three, the tensor-rank of very specific tensors like the
matrix multiplication tensor~\cite{Blaser1999, Shpilka2003}, the
finite field multiplication tensor~\cite{ChudnovskyC1988,ShparlinskiTV1992} and the
polynomial multiplication tensor~\cite{BrownD1980,Kaminski2005} has
been studied in prior works. For this case, the current best lower
bound known for any explicit tensor over $\F_2$ is a lower bound of
$3.52k$ for the finite field multiplication tensor due to Chudnovsky
and Chudnovsky~\cite{ChudnovskyC1988,ShparlinskiTV1992}, which builds
on the lower bound result of Brown and Dobkin~\cite{BrownD1980} for
the polynomial multiplication tensor. For general fields, the best
known lower bound for any explicit tensor is $2.5k - o(k)$ for the
matrix multiplication tensor due to Bl\"{a}ser~\cite{Blaser1999}. 

 Also relevant to this discussion is a recent result of Effremenko et al.~\cite{EGOW17}, who showed that a fairly general class of lower bound techniques called \emph{rank methods} are not strong enough to give lower bounds on tensor rank stronger than $2^{d}\cdot k^{\lfloor d/2 \rfloor}$. In a nutshell, not only can we not prove good tensor rank lower bounds, we do not even have techniques, which `in principle' could be useful for such lower bounds! 
 
\subsection{Our results}
We make contributions to both the above questions by studying {\em multilinear
forms} and their {\em bias}. A $d$-linear form is a map $f : (\F_2^k)^d \to \F_2$ which is linear in each of its arguments.
The {\em bias} of a $d$-linear form is  defined as follows.
$$ \bias(f) := \left| {\mathbb E}_{x_1, \ldots, x_d \in \F_2^k} [ (-1)^{f(x_1, \ldots, x_k)}] \right| \, .$$
This  measures the difference between the probability of output $1$
and output $0$. Similarly, the correlation of a $d$-linear form $f$
with another function $g$ is defined as $\corr(f,g) := \bias(f-g)$,
which measures the difference between the probabilities (on a random
input) that $f$ and $g$ agree and disagree.

A $d$-linear form $f$ can naturally be viewed as a polynomial of degree $d$
in $n = kd$ variables. We can then ask, for some $\ell  d$,
is there a $d$-linear form $f$ such that the correlation of $f$ with every degree $\ell$ polynomial
in $\F_2[X_1, \ldots, X_n]$ is small?
Knowing the existence of a $d$-linear $f$ that achieves this small correlation property 
gives a significantly reduced search space for finding an explicit
$f$ with small correlation with lower degree polynomials. Our first result gives a positive answer to this question for a large range of $\ell$ and $d$.

\medskip
\noindent {\bf Theorem A (informal). \ }
{\it Let $d \ll o(n/log n)$ and let $k = \frac{n}{d}$.
Let $\ell < d/2$. Then with high probability, for a uniformly random $d$-linear form
$f: (\F_2^k)^d \to \F_2$, we have that for all polynomial $P(X_1, \ldots, X_n) \in \F_2[X_1, \ldots, X_n]$
of degree at most $\ell$:
$$ \corr(f, P) \leq 2^{-k ( 1-o(1))} = 2^{-\frac{n}{d}(1-o(1))}.$$

Moreover, for every $d$-linear form, there is a degree $0$ polynomial $P$ (namely the constant $0$ polynomial)
such that $$\corr(f, P) \geq \Omega(2^{-k}).$$
}
\medskip

For $d$ small enough ($\tilde{O}(\log n)$), the above theorem actually holds with $\ell = d-1$.

An important step towards proving Theorem A is a precise understanding of the distribution of the
{\em bias} of a random $d$-linear form. Along the way, we give tight upper bounds on the probability that the sum of $t$ random {\em rank-1} $d$-dimensional tensors equals $0$.

Previously, a beautiful result of Ben-Eliezer, Lovett and Hod~\cite{Ben-EliezerHL2012} showed that for
all $d < \alpha n$, there are polynomials $f(X_1, \ldots, X_n)$ of 
degree $d$ whose correlation with polynomials of degree $\ell = d - 1$ is
$2^{-\Omega(n/d)}$. The results are incomparable; the $f$ in~\cite{Ben-EliezerHL2012} need not come from a $d$-linear form,
and for this more general setting the bound $2^{-\Omega(n/d)}$ might not be tight,
but on the positive side~\cite{Ben-EliezerHL2012} can handle larger $d$ while proving correlation bounds against
polynomials with degree as large as $d-1$.

A $d$-linear form $f$ can also be naturally viewed as a $d$-dimensional tensor.
Indeed, $f$ can be completely specified by the tensor $T$ of values $f(e_{i_1}, e_{i_2}, \ldots, e_{i_d})$, as the $i_j$
vary in $[k]$. We can then ask, are there natural properties of the $d$-linear form $f$ 
which would imply that the tensor rank of $T$ is high?

We show that having low bias, which is a simple measure of pseudorandomness for $d$-linear forms, already implies something nontrivial about the tensor rank. We prove a lower bound on the tensor rank in terms of the bias of the form.

\medskip
\noindent {\bf Theorem B. \ }
{\it Let $f : (\F_2^k)^d \to \F_2$ be a $d$-linear form.
Let $T$ be its associated tensor, and let $t$ be the rank of $T$.
Then  $$\bias(f)\geq \left(1 - \frac{1}{2^{d-1}}\right)^t.$$  	

In particular, if $\bias(f) = 2^{-(1-o(1))k}$, then 
$$t \geq  k \cdot \log_2 \frac{2^{d-1}}{2^{d-1}-1}.$$

Moreover, for every $t$ there is a tensor $T$ with
tensor rank $t$ such that the following is true.
$$ \bias(f)\leq \left(1 - \frac{1}{2^{d-1}}\right)^t + \frac{d}{2^k}.$$
}
\medskip

This lower bound on tensor rank in terms of bias is almost optimal for any fixed $d$. It implies that any explicit $d$-linear form with low bias (such $d$-linear forms are easy to construct) automatically must have tensor rank $(1 + \Omega(1)) \cdot k$. Purely from the point of view of proving tensor rank lower bounds for explicit tensors, these results are only interesting in the case of $d = 3$ (for larger $d$ the implied tensor rank lower bounds fail to beat trivial explicit tensor rank lower bounds).

For $d=3$, this gives a natural and clean route to proving nontrivial
tensor rank lower bounds for explicit tensors. In particular,
trilinear forms with nearly minimal bias of of $2^{-(1-o(1))k}$ must have tensor rank at least $2.409k$ (which happens to be tight). 
A finer analysis of our arguments shows that
trilinear forms with {\em exactly} minimal bias of $\approx 2 \cdot 2^{-k}$, such as the
finite field multiplication tensor, have tensor rank $\geq 3.52 k$,
thus matching the best known explicit tensor rank lower bound for $3$-dimensional tensors~\cite{BrownD1980,ChudnovskyC1988,ShparlinskiTV1992}.
It also immediately implies that the matrix multiplication tensor has tensor rank $\geq 1.8 k$, which is nontrivial (but still far from the best known bound of
$3k$~\cite{Shpilka2003, Blaser1999}).

\subsection{Methods}

Underlying our main results, Theorem A and Theorem B, are 
two related combinatorial bounds involving rank-$t$ $d$-linear forms.
We now state these bounds for the special case of $d=3$.
For $i \in [t]$, let $x_i, y_i, z_i \in \F_2^k$. Let
$P_i(u, v, w)$ be the trilinear form defined as
$$ P_i (u,v,w) = \langle u, x_i \rangle \cdot \langle v, y_i \rangle \cdot \langle w, z_i \rangle.$$
Now, consider the trilinear form $P(u,v,w)$ given by
$$ P(u,v,w) = \sum_{i=1}^t P_i (u,v,w).$$
Then,  we have the following.
\begin{enumerate}
\item If $x_i, y_i, z_i$ are picked uniformly at random from $\F_2^k$,
then the probability that $P$ is identically $0$ is very small.
Concretely,
$$ \Pr_{x_i, y_i, z_i} [ P \equiv 0 ]$$
is about $2^{-kt}$, provided $t \ll k^2$.
This bound is essentially optimal.

\item  For arbitrary $x_i, y_i, z_i$, 
the bias of $P$ is large.
Concretely,
$$ \min_{x_i, y_i, z_i} [ \bias(P) ] \geq (3/4)^t.$$
This bound is also essentially optimal.
\end{enumerate}

We now give an outline of the proofs of Theorem A and Theorem B.

The proof of Theorem A follows the high-level outline of~\cite{Ben-EliezerHL2012}. We first use the method of moments
to show that for a fixed $n$-variate polynomial $P$ of degree $\ell$, the correlation of a random $d$-linear $f$ with $P$ is
small with extremely high probability. Then, by a union bound over all $P$, we
conclude that a random $f$ is uncorrelated with all $P$ with quite high probability.

Implementing this approach gives rise to some natural and interesting questions about
rank-1 tensors. How many rank-1 tensors can lie in a given low dimensional linear space
of tensors? Given a collection of $t$ random rank-1 tensors, what is the probability that
the dimension of the space spanned by them is small? What is the probability that the sum
of $t$ random rank-1 tensors equals $0$? We investigate these questions using linear-algebraic
ideas, and obtain near-optimal answers for all of them.

For example, the $d = 3$ case requires us to study the probability that
$$ \sum_{i=1}^t x_i \otimes y_i \otimes z_i = 0.$$
By some simple manipulations, this reduces to bounding the probability that
the linear space of matrices
$$ span \{ x_i \otimes y_i  : i \in [t] \} $$
has dimension $\leq t-r$. 
We bound this by studying the probability that $x_i \otimes y_i$ lies
in the linear space 
$$ span \{ x_j \otimes y_j  : j \in [i-1] \}.$$
This final probability is bounded using the following general theorem.

\medskip
\noindent {\bf Lemma. \ }
{\it
For any linear space $U \subseteq \F_2^{k^2}$ of dimension $u \ll k^2$, the probability 
that $x \otimes y \in U$ is at most $\tilde{O}\left(\frac{2^{u/k}}{2^{k}}\right)$.
}
\medskip

The proof of this lemma is hands on, and uses basic linear algebra and 
some elementary analytic inequalities. The key is to take an echelon form basis for $U$. We use this basis to understand which $\tilde{x}\in \F_2^k$ are ``important"; i.e., they have the property
that $\tilde{x} \otimes y \in U$ with noticeable probability for a random $y$.

The above lemma is essentially tight: with $U = V \otimes \F_2^k$ and
$\F_2^k \otimes V$ being tight examples. The sets of
the important $\tilde{x}$ in these two examples look very different. Because of
this, our final proof involves proving tight upper bounds on an analytic
maximization problem that has multiple very different global maxima.

For Theorem B, which gives a relationship between tensor rank and bias, the proof proceeds in the contrapositive. We show that any $d$-linear form whose underlying tensor has low rank must have high bias. Let us illustrate
the underlying ideas in the case of $d = 3$.
Here, we are given the $3$-linear form $P$, defined as 
$$ P(u,v,w) = \sum_{i=1}^t \langle x_i, u \rangle\cdot \langle y_i, v \rangle \cdot \langle z_i, w \rangle.$$
We want to show that this has high bias if $t$ is small.
The key claim that we show is the following.

\medskip
\noindent {\bf Lemma. \ }
{\it Let $y_1, \ldots, y_t, z_1, \ldots, z_t \in \F_2^t$.
For at least $(3/4)^t$ fraction of the pairs  $(v,w) \in \F_2^t$, we have
that for all $i \in [t]$:
$$\langle v, y_i \rangle \cdot \langle w, z_i \rangle = 0.$$}
\medskip

For any fixed $i$, the set of $(v,w)$ satisfying the above
is the union of two codimension $1$ hyperplanes in $\F_2^{2t}$,
and thus a random $(v,w)$ satisfies it with probability $3/4$.
The above lemma shows that the probability of all
these events happening together is at least as large as it would have been
had they been independent.

\section{Preliminaries}
Unless otherwise stated, we always work over the field $\F_2$. We use capital $X, Y, Z$ etc.\ to denote formal variables or sets of formal variables, and small letters $x, y, z$ to denote instantiations of these formal variables. 

For integers $n,d\geq 0$, denote by $\poly(n,d)$ the set of all degree $\leq d$ multilinear polynomials in $\F_2[X]$, where $X=\{X_1,...,X_n\}$ is a variable set. Note that every $f\in \poly(n,d)$ naturally corresponds to a unique map $f:\F_2^n \to \F_2$. 
\subsection{Bias and Correlation}
Two fundamental notions used in this paper are those of bias and correlation, which we now define.
\begin{definition}[Bias]
Bias of a function $f:\F_2^n\to \{0,1\}$ is defined as 
$$
\bias(f):= \left| \E_{x\in \F_2^n} (-1)^{f(x)}\right|.
$$
The bias of an $\F_2$-valued function $f:\F_2^n\to \F_2$ is defined as $\bias(f):=\bias(\iota(f))$, where $\iota$ is the standard map from $\F_2$ to $\{0,1\}$. 
\end{definition}

\begin{definition}[Correlation]
We define the correlation between two functions $f,g: \F_2^n\to \F_2$, by 
$$
\corr(f,g):= \bias(f-g) \, .
$$ 
\end{definition}

Given a function $f:\F_2^n\to \F_2$, we will be interested in its maximum correlation with low degree polynomials. Towards this we define 
$$
\corr(f,d):= \max_{g\in \poly(n,d)} \corr(f,g) \, . 
$$

More generally, given a class $\mathcal{C}$ of functions, we define
$$
\corr(f,\mathcal{C}):= \max_{g\in \mathcal{C}} \corr(f,g) \, .
$$

\subsection{Tensors and $d$-linear forms}

Tensors are  generalizations of matrices to higher dimensions.
\begin{definition}[Tensors and Tensor rank]\label{def:tensor rank}
Let $k$ and $d$ be natural numbers.  A $d$ dimensional tensor $T$ of size $k$ over a field $\F$ is a map $T: [k]^d \rightarrow \F$. $T$ is said to be of rank one if there exist $d$ vectors $u_1, u_2, \ldots, u_d : [k] \rightarrow \F$ such that for every $(i_1, i_2, \ldots, i_d) \in [k]^d$, $T(i_1, i_2, \ldots, i_d) = \prod_{j = 1}^d u_j(i_j)$. The rank of $T$ is the minimum $t$ such that $T$ can be written as a sum of $t$ rank one tensors. 
\end{definition}
Every matrix can be naturally associated with a bilinear polynomial, and in some cases, one can study the properties of this bilinear polynomial as a proxy of studying various properties of the matrix itself. This paradigm also generalizes to tensors, as the following definition indicates.
\begin{definition}[Tensors as Multilinear Forms] 
Let $T : [k]^d \rightarrow \F$ be a $d$ dimensional tensor. Then, the \emph{set-multilinear} polynomial associated with $T$ is the polynomial $f_T$ in variables $\set{X_{i,j} : i \in [d], j \in [k]}$ over $\F$ defined as follows. 
\[
f_T(X_{1,1}, X_{1,2}, \ldots, X_{d, k}) = \sum_{(i_1, i_2, \ldots, i_d) \in [k]^d} T(i_1, i_2, \ldots, i_d)\cdot \prod_{j = 1}^d X_{j, i_j}.
\]
\end{definition}
Given the above association between $d$-dimensional tensors and
$d$-linear forms, we will use the terms tensor and $d$-linear form
interchangeably.  
\subsection{Some explicit tensors}
We now define some explicit tensors which we use at various places in this paper. We start with the trace function.
\subsubsection{Trace tensor}

\begin{definition}\label{def:trace map}
$\tr:\F_{2^k} \rightarrow \F_2$ is the $\F_2$-linear map defined as follows. 
\[
\tr(\alpha) = \alpha + \alpha^2 + \ldots + \alpha^{2^{k-1}} \, .
\]
 \end{definition}
The $\tr$ map will be useful for us as we define the candidate hard tensor for our lower bounds.
\begin{definition}\label{def:trace polynomial}
Let $Tr:\F_2^{k\times k\times k} \rightarrow \F_2$ be the function defined as follows.
$$
Tr(X,Y,Z) := \tr(XYZ),
$$
where $XYZ$ denotes multiplication over the larger field $\F_{2^k}$ when $X = (X_1, X_2, \ldots, X_k), Y = (Y_1, Y_2, \ldots, Y_k), Z = (Z_1, Z_2, \ldots Z_k)$ are viewed as encodings of elements in $\F_{2^k}$.
\end{definition}
Since $\tr$ is an $\F_2$-linear map, the function $Tr(X, Y, Z)$ can be
viewed as a $3$-linear polynomial in the variables $X = (X_1,
X_2, \ldots, X_k), Y = (Y_1, Y_2, \ldots, Y_k), Z = (Z_1, Z_2, \ldots
Z_k)$. For the rest of this paper, when we say $Tr(X, Y, Z)$, we refer
to this natural $3$-linear polynomial and the three dimensional
tensor associated with it.  We remark that, upto change of basis, this
is the finite field multiplication tensor, which was analyzed by
Chudnovsky-Chudnovsky~\cite{ChudnovskyC1988} and Shparlinksi-Tsfasman-Vladut~\cite{ShparlinskiTV1992}.

\subsubsection{Matrix multiplication tensor}
\begin{definition}\label{def:mm tensor}
The tensor corresponding to the product of two $n\times n$ matrices is defined as 
\[
M_n(X, Y, Z) = \sum_{i  = 1}^n \sum_{j = 1}^n \sum_{k = 1}^n X_{i,j} Y_{j, k} Z_{i,k} \, .
\]
Here, $X = \set{X_{i,j} : i, j \in [n]}, Y = \set{Y_{i, j} : i, j \in [n]}, Z = \set{Z_{i, j} : i, j \in [n]}$.
\end{definition}
Note that $M_n(X,Y,Z)$ is the trace of the matrix product $X\cdot Y
\cdot Z^T$. In other words, $M_n(X,Y,Z^T) = \tr(X\cdot Y \cdot
Z)$. Note this is the matrix trace and is different from the trace function considered in the
previous section where we viewed $X, Y, Z$ as elements of the large field. 


\section{Correlation of random $d$-linear forms}
In this section, we study the correlation of random $d$-linear forms 
with lower degree polynomials.\\
Our main result in this section is the following theorem, which 
states that a random $d$-linear form is uncorrelated with 
degree $\ell$ polynomials under certain conditions.

\begin{theorem}
\label{thm:smallcorr}
Let $\ell, d, n$ be integers such that
$d \mid n$,  $d = o(\frac{n}{\log n})$ and $\ell < d/2$.
Set $k = n/d$.


Pick a uniformly random $d$-linear form $f: (\F_2^k)^d \to \F_2$.
Then, with probability $1-o(1)$, $f$ has the following property.
For all polynomials $P(X_1, \ldots,X_n) \in \F_2[X_1, \ldots, X_n]$
with degree at most $\ell$, we have,
$$ \corr(f, P) < 2^{-(1-o(1))n/d } \, .$$
\end{theorem}

Along the way, we develop several tools to understand the bias of 
random $d$-linear forms. For example, we show that a random
$d$-linear form is unbiased with extremely high probability.

\begin{theorem}
\label{thm:smbias}
Let $\epsilon > 0$ be fixed.
Let $d, k$ be integers with $d <  2^{\epsilon k/5}$, and consider a uniformly random
$d$-linear form $ f: (\F_2^k)^d \to \F_2$.
Then, 
$$\Pr[\bias(f) \ge 2^{-(1-\epsilon)k}] \le 2^{-\Omega(\epsilon^2 k^d)}\, .$$
\end{theorem}

\begin{remark}
Note that any $d$-linear form $f(X_1,\dots,X_d)$ vanishes if any one of the block of
variables $X_1,\dots,X_d$ is zero. Hence, the bias of any
$d$-linear form (or equivalently its correlation with the constant 0
polynomial) is at least $2^{-k} =
2^{-n/d}$. \cref{thm:smbias} states that it is extremely unlikely for
a random $d$-linear form to have even slightly more bias while
\cref{thm:smallcorr} states that it is extremely unlikely for a random
$d$-linear form to have slightly better correlation with any
degree $\ell$ polynomial.
\end{remark}

The key ingredient in the proofs of the above theorems is 
the following theorem on the distribution of the sum
of random rank-$1$ tensors.

\begin{theorem}
\label{thm:tensorsumbound}
Let $\epsilon > 0$ be a constant.
 Let $d, k, t$ be integers with $d < 2^{\epsilon k/5}$, and $t < \frac{\epsilon}{5} k^{d-1}$. Let $\{x^{(i,j)}\}_{i \in [t], j \in [d]}$ be picked independently and uniformly distributed in $\F_2^k$.Then, 
 $$ \Pr\left[\sum_{i = 1}^t \bigotimes_{j =1}^d x^{(i,j)} = 0\right] \le 2^{-(1-\epsilon/2) \cdot kt}.$$
\end{theorem}

\begin{remark}
If any block of vectors (say wlog. $\{x^{(i,1)}\}_{i\in[t]}$, the
first block of vectors) are all $\overline{0}$ (this happens with
probability $2^{-kt}$), then the $d$-dimensional linear form $\sum_{i
  = 1}^t \bigotimes_{j =1}^d x^{(i,j)}=0$. The above theorem states
that the probability of the $d$-linear form vanishing is not
significantly larger.
\end{remark}

In turn, the proof of the above theorem is based on the following lemma, which gives an upper bound
on the probability that a random rank-$1$ tensor  lies in a fixed low dimensional subspace.

\begin{lemma} 
\label{lem:randrankone}
 Let $k,d$ be integers and $U$ be a subspace of $(\F_2^k)^{\otimes d}$ of dimension $u$.
Let $x_1, \ldots, x_d \in \F_2^k$ be picked independently and uniformly at random,
and let $T = \otimes_{i=1}^d x_i$. 
Then, 
 $$\Pr[T \in U] \le  \frac{d}{2^k} + \frac{2^{u/k^{d-1}}}{2^{k}} \, .$$
\end{lemma}

\begin{remark}
Let $U = V \otimes (\F_2^k)^{\otimes (d-1)}$ where $V$ is a
$u/k^{d-1}$-dimensional subspace of $\F_2^k$. Note, $\dim(U) =
u$. Clearly, $\Pr[\otimes_{i=1}^d x_i \in U] = \Pr[x_1 \in V] =
2^{u/k^{d-1}}/2^k$. The above lemma states that the probability is not
significantly larger than this for any other $U$.
\end{remark}


In the next subsection, we show how \cref{thm:smallcorr} and \cref{thm:smbias} follow from \cref{thm:tensorsumbound}. After that, we prove
\cref{thm:tensorsumbound} by studying the distribution of the dimension of a collection of random rank $1$ tensors.

\subsection{Proofs of {\cref{thm:smallcorr}} and {\cref{thm:smbias}}}

We first prove \cref{thm:smbias}.
\begin{proof}[Proof of \cref{thm:smbias}]
We want to bound $\Pr_f[\bias(f) \ge 2^{-(1-\epsilon)k}]$. We shall do so by bounding the $t^{th}$ moment of $\bias(f)$ for a suitable choice of $t$ and applying Markov's inequality.

Let $T : [k]^d \to \F_2$ denote the tensor associated with $f$. Thus $T(i_1, \ldots, i_d)$ are all independent and uniformly distributed in $\F_2$.

We now compute the $t^{th}$ moment of $f$.
\begin{align*}
\E_f[(\bias(f))^t] &= \E_f \left[ \left( \E_{x^{(1)}, \ldots, x^{(d)} \sim \F_2^k}\left[ (-1)^{f\left(x^{(1)}, \ldots, x^{(d)}\right)}\right]\right)^t \right]\\
&= \E_f \left[ \prod_{i \in [t]} \left( \E_{x^{(i,1)}, \ldots, x^{(i,d)} \sim \F_2^k} \left[ (-1)^{f\left(x^{(i,1)}, \ldots, x^{(i,d)}\right)}\right]\right) \right]\\
&= \E_{\{x^{(i,j)}\}_{i \in [t], j \in [d]}} \left[ \E_{f} \left[ (-1)^{\sum_{i = 1}^t f\left(x^{(i,1)}, \ldots, x^{(i,d)}\right)} \right] \right]\\
&= \E_{\{x^{(i,j)}\}_{i \in [t], j \in [d]}} \left[ \prod_{(\ell_1, \ldots, \ell_d) \in [k]^d}\left( \E_{T(\ell_1, \ldots, \ell_d) \sim \F_2} \left[ (-1)^{T(\ell_1, \ldots, \ell_d)\cdot\left( \sum_{i =1}^t \prod_{j=1}^d x^{(i,j)}_{\ell_j} \right)}\right] \right) \right]\\
&= \E_{\{x^{(i,j)}\}_{i \in [t], j \in [d]}} \left[ \prod_{(\ell_1, \ldots, \ell_d) \in [k]^d} \mathbbm{1}_{\sum_{i =1}^t \prod_{j=1}^d x^{(i,j)}_{\ell_j} = 0} \right]\\
&= \E_{\{x^{(i,j)}\}_{i \in [t], j \in [d]}} \left[  \mathbbm{1}_{\forall (\ell_1, \ldots, \ell_d) \in [k]^d,\ \sum_{i =1}^t \prod_{j=1}^d x^{(i,j)}_{\ell_j} = 0} \right]\\
&= \Pr_{\{x^{(i,j)}\}_{i \in [t], j \in [d]}}\left[\forall (\ell_1, \ldots, \ell_d) \in [k]^d,\ \sum_{i =1}^t \prod_{j=1}^d x^{(i,j)}_{\ell_j} = 0\right]\\
&= \Pr_{\{x^{(i,j)}\}_{i \in [t], j \in [d]}}\left[\sum_{i = 1}^t \bigotimes_{j =1}^d x^{(i,j)} = 0\right].
\end{align*}
Setting $t = \frac{\epsilon}{10} k^{d-1}$, \cref{thm:tensorsumbound} tells us that 
$$\E_f[(\bias(f))^t]=2^{-\left(1 - \epsilon/2 \right) kt}.$$
Using Markov's inequality,
\begin{align*}
\Pr_f\left[\bias(f) \ge 2^{-(1-\epsilon)k}\right] \le
                                                    \frac{2^{-(1-\epsilon/2))kt}}{2^{-(1-\epsilon)kt}}\le
                                                                                                        2^{-\epsilon
                                                                                                        kt/2}
  \le 2^{- \Omega(\epsilon^2 k^d)}
\end{align*}
as claimed.

\end{proof}

We now use a similar argument to prove \cref{thm:smallcorr}.
\begin{proof}[Proof of \cref{thm:smallcorr}]
Fix an arbitrary $\epsilon > 0$.
Let $\mathcal C$ denote the space of degree $\leq \ell$ polynomials in
$\F_2[X_1, \ldots, X_n]$. We want to show that with high probability over
the choice of $f$, we have that for every $P \in \mathcal C$,
$ \corr(f, P) \leq 2^{-(1-\epsilon)k}$.

Fix $P \in \mathcal{C}$ and consider the $t^{th}$ moment of $\bias(f - P)$. Imitating the proof of \cref{thm:smbias}, we get
\begin{align*}
\E_f[(\bias(f - P))^t] &= \E_{\{x^{(i,j)}\}_{i \in [t], j \in [d]}} \left[ (-1)^{\sum_{i=1}^t P\left(x^{(i,1)}, \ldots, x^{(i,d)}\right)} \cdot\mathbbm{1}_{\forall (\ell_1, \ldots, \ell_d) \in [k]^d,\ \sum_{i =1}^t \prod_{j=1}^d x^{(i,j)}_{\ell_j} = 0} \right]\\
&\le \E_{\{x^{(i,j)}\}_{i \in [t], j \in [d]}} \left[ \mathbbm{1}_{\forall (\ell_1, \ldots, \ell_d) \in [k]^d,\ \sum_{i =1}^t \prod_{j=1}^d x^{(i,j)}_{\ell_j} = 0} \right]\\
&= \Pr\left[\sum_{i = 1}^t \bigotimes_{j =1}^d x^{(i,j)} = 0\right].
\end{align*}
Now we will apply \cref{thm:tensorsumbound}.
Observe that since $d = o(n/\log n)$, we have,
$$  d < 2^{\epsilon k / 5}.$$
As in the proof of \cref{thm:smbias}, we set $t = \frac{\epsilon}{10}k^{d-1}$, 
invoke \cref{thm:tensorsumbound} and apply Markov's inequality to get,
$$\Pr_f\left[\bias(f - P) \ge 2^{-(1-\epsilon)k}\right] \le 2^{-\epsilon^2 k^d/20}.$$

Now $\bias(f-P) = \corr(f,P)$.
Thus, by a union bound over all $P \in \mathcal C$, we have the following.
\begin{align}
\label{eqsmcorr}
\Pr_f\left[\corr(f, \mathcal{C}) \ge 2^{-(1-\epsilon)k}\right] \le |\mathcal{C}| \cdot 2^{-\epsilon^2 k^d/20}.
\end{align}

It remains to estimate $|\mathcal C|$.  We show below that $|\mathcal
C| = o(k^d)$. The proof of this lemma works for any other $C$ as long
as $C$ satisfies $|\mathcal
C| = o(k^d)$.
Note that $|\mathcal C| = 2^{{n \choose \leq \ell}}$.
Let $\delta$ denote $d/n$.
\begin{align*}
{n \choose \leq \ell} &\leq {n \choose \leq d/2} \leq \left( \frac{2e n}{d} \right)^{d/2} \leq \left( \frac{2e}{\delta} \right)^{\delta n/2}\\
&= o\left(\left( \frac{1}{\delta} \right)^{\delta n}\right) \quad \quad
  [\text{Since } \delta  = o(1)] \\
& = o( k^d ).
\end{align*}

Combining this with Equation~\eqref{eqsmcorr}, we get,
$$\Pr_f\left[\corr(f, \mathcal{C}) \ge 2^{-(1-\epsilon)k}\right] \le 2^{o(k^d)} \cdot 2^{-\epsilon^2 k^d/20}.$$

Since this holds for every $\epsilon > 0$, we get the desired result.
\end{proof}


\subsection{Random rank-1 tensors}

In this subsection, we first prove~\cref{lem:randrankone} on the probability that a random rank-$1$ tensor
lies in a fixed low-dimensional subspace. We then give a corollary of this lemma which bounds the probability
that a collection of random rank-$1$ tensors spans a very low dimensional subspace. This corollary will be used
in the proof of~\cref{thm:tensorsumbound}.

\begin{proof}[Proof of {\cref{lem:randrankone}}]
Define $$f_{d,k}(u) =  \left( 1 - (1-\frac{1}{2^k})^{d-1} \right)  + (1-\frac{1}{2^k})^{d-1} \cdot \frac{2^{u/k^{d-1}}}{2^{k}}.$$
We will prove, by induction on $d$, the following stronger bound. 
 $$\Pr[T \in U] \le  f_{d,k}(u).$$
The fact that this implies the lemma, follows from the observations that $1-\frac{d-1}{2^k} \leq (1-\frac{1}{2^k})^{d-1}$ and that $(1-\frac{1}{2^k})^{d-1} \leq 1$.
\paragraph*{Base case. } The $d=1$ case is trivial (using the observation that $f_{1,k}(u) = \frac{2^u}{2^k}$). We now show the statement holds for larger $d$.
\paragraph*{Induction step. }
Let $k' = k^{d-1}$. We will view $(\F_2^k)^{\otimes d}$ as $\F_2^k \otimes \F_2^{k'}$.
Every element $v$ of $(\F_2^k)^{\otimes d}$ can thus be written as a tuple $(v_1, \ldots, v_k)$, where each $v_i$
is an element of $\F_2^{k'}$ (thus the $k^d$ coordinates are partitioned into $k$ blocks of coordinates,
with each block having $k'$ coordinates).
We let $\pi_i : (\F_2^{k})^{\otimes d} \to \F_2^{k'}$ be the $i$th projection map, mapping $v$ to $v_i$.

With this convention, we take a basis for $U$ in \textit{row echelon form}.
Concretely, this gives us a basis $\mathcal B$ for $U$,
such that $\mathcal B$ is a disjoint union of $\mathcal B_1, \ldots, \mathcal B_k$ ($\mathcal B_j$ is
the set of basis vectors pivoted in the $j$'th block of coordinates),
such that,
\begin{itemize}
\item for all $v \in \mathcal B_j$ and $i <j$, $\pi_i(v) =  0$,
\item the vectors $\pi_j(v) \in \F_2^{k'}$, as $v$ varies in $\mathcal B_j$, are linearly independent.
\end{itemize}
Define $U_j = \vspan\{ \pi_{j}(v) \mid v \in \mathcal{B}_j) \}$. Thus we have
$\dim(U_j) = |\mathcal B_j|$ and
$$ \sum_{j=1}^k \dim(U_j) = \dim(U).$$
For $i > j$, we define a linear map $\psi_{ij} : U_j \to \F_2^{k'}$ by
defining $\psi_{ij}$ on a basis for $U_j$:
$$ \psi_{ij}( \pi_j(v) ) = \pi_i(v),\ \forall v \in B_j.$$
Then we have the following basic claim (which follows immediately from the above echelon form representation of $U$).
\begin{claim}
\label{claim:rowform}
Let $v \in (\F_2^k)^{\otimes d}$. Then $v \in U$ only if there exists $(u_1, \ldots, u_k) \in \prod_{i = 1}^k U_i$ such that for each $i \in [k]$ we have
$$\pi_i(v) = u_i + \sum_{j < i} \psi_{ij} (u_j).$$
\end{claim}
To simplify notation, we will denote $x_1$ by $y$
and $\otimes_{i=2}^{d} x_i$ by $z$. We want to find
an upper bound on $\Pr[y \otimes z \in U]$. 
\begin{claim}
Let $\tilde{z} \in (\F_2^k)^{\otimes (d-1)}$ and $S = \{i\ |\ \tilde{z} \in U_i\}$, then, 
$$ \Pr_{y \in \F_2^k}[ y \otimes \tilde{z} \in U] \leq \frac{2^{|S|}}{2^k}.$$
\end{claim}
\begin{proof}
For fixed $\tilde{z}$, given the random variable $v = y \otimes \tilde{z}$,
 we define random variables $u_1, u_2, ... , u_k$ by:
$u_i  := \pi_i (v) - \sum_{j < i} \psi_{ij}(u_j)$. Note that $\pi_i(v)=\pi_i(y
\otimes \tilde{z}) = y_i \tilde{z}$. Also note that $u_i$ is only a
function of $y_1,\dots,y_i$. By \cref{claim:rowform}, $v \in
U$ only if for all $i$, $u_i \in U_i$.
\begin{align*}
\Pr_{y \in \F_2^k}[ y \otimes \tilde{z} \in U]
&\le   \Pr_{y} \left[ \forall i\leq k,\;  u_i \in U_i \right]\\
&= \prod_{i=1}^k \Pr \left[   u_i \in U_i \;\middle|\; u_1 \in U_1, \ldots, u_{i-1} \in U_{i-1} \right]\\
&= \prod_{i=1}^k \E_{u_1 \in U_1,\ldots,u_{i-1} \in U_{i-1}}\left[\Pr_{u_i}\left[  u_i
\in U_i \;\middle|\; u_1,\ldots, u_{i-1} \right] \right] \\
&= \prod_{i=1}^k \E_{u_1 \in U_1, \ldots, u_{i-1} \in U_{i-1}} \left[Pr_{u_i} \left[
\pi_i(v)  - \sum_{j < i} \psi_{ij} (u_j) \in U_i \;\middle|\; u_1,\ldots ,
u_{i-1} \right]\right]\\
&= \prod_{i=1}^k \E_{u_1 \in U_1,\ldots, u_{i-1} \in U_{i-1}}\left[Pr_{y_i}\left[   y_i
\tilde{z}  - \sum_{j < i} \psi_{ij} (u_j) \in U_i \;\middle|\; u_1 , \ldots ,u_{i-1} \right]\right]\\
& \leq\prod_{i \not\in S}   \left(\frac12\right) =
  \left(\frac12\right)^{k-|S|},
\end{align*}
where the last inequality follows since for every $i \notin S$ and
every vector $w$, at most one of $w$ and $w + \tilde{z}$ can lie in
$U_i$ (as $\tilde{z} \notin U_i$). 

\end{proof}
For $S \subseteq [k]$, let
$$ U_S = \bigcap_{i \in S} U_i \, .$$
Then,
\begin{align*}
\Pr_{y, z}[y \otimes z \in U] &\leq \mathbb E_{z} \left[ \frac{2^{\sum_{i=1}^k 1_{U_i}(z)}}{2^k}  \right] \quad \quad
  \text{[Follows from the above claim]}  \\
&= \frac{1}{2^k} \mathbb E_{z} \left[ \prod_{i=1}^k  2^{1_{U_i}(z)} \right]\\
&= \frac{1}{2^k}\mathbb E_{z} \left[ \prod_{i=1}^k  (1 + 1_{U_i}(z)) \right]\\
&= \frac{1}{2^k} \mathbb E_{z} \left[ \sum_{S \subseteq [k]} 1_{U_S}(z) \right] \\
&= \frac{1}{2^k} \sum_{S \subseteq [k]}  \Pr_z [ z \in U_S]. \\
\end{align*}
Now, observe that for each $i \in S$, we have $\Pr[z \in U_S] \leq \Pr[z \in U_i]$.
Thus if we sort the $U_i$ so that $\dim(U_1) \geq \dim(U_2) \geq \ldots \geq \dim(U_k)$,
then we have the following sequence of inequalities.
\begin{align*}
 \Pr_{y,z} [ y \otimes z \in U] &\leq \frac{1}{2^k} \left( 1 +  \sum_{i \in [k]} \sum_{S \subseteq [i], i \in S}  \Pr_z [ z \in U_S] \right)\\
&\leq \frac{1}{2^k} \left( 1 + \sum_{i \in [k]}  2^{i-1} \Pr_z [ z \in U_i]  \right)\\
&\leq \frac{1}{2^k}\left( 1 +  \sum_{i \in [k]}  2^{i-1} f_{d-1, k}(\dim(U_i)) \right),\\
\end{align*}
where the last step follows from the induction hypothesis. To find an upper bound for this last expression, we let $a_i = \dim(U_i)$.
We have the constraints
$$ \sum_i a_i = u,$$
$$ k' \geq a_1 \geq a_2 \geq \ldots \geq a_k \geq 0,$$
where $k' = k^{d-1}$, and we want to maximize an expression of the form
$$ \sum_{i=1}^k 2^{i-1} (\alpha + \beta 2^{a_i/k^{d-2}} ) = \alpha\cdot (2^{k} - 1) + \beta\cdot \left( \sum_{i = 1}^k 2^{i - 1 + a_i/k^{d-2}}\right) .$$
where $\alpha,\beta > 0$.

It is worth noting what happens in the two examples $U = V \otimes \F_2^{k'}$
and $U = \F_2^{k} \otimes W$, where $V \subseteq \F_2^{k}$ and $W \subseteq 
\F_2^{k'}$ are subspaces of the appropriate dimension. In the first case,
$a_1 = a_2 = \ldots = a_{u/k'} = k'$ and the remaining $a_i$ are $0$.
In the second case, all the $a_i = u/k$. Both are global maxima of
the expression we want to maximize! The existence of these very different
maxima makes this maximization problem somewhat tricky.

In~\cref{thm:max} we prove a tight upper bound for this function. For every $i \in [k]$, let $b_i = a_i/k^{d-2}$, and let $\tilde{u} = u/k^{d-2}$. Then, $b_1, b_2, \ldots, b_k$ and $\tilde{u}$ satisfy the constraints in the hypothesis of~\cref{thm:max}, and~\cref{thm:max} tells us that a global maxima is achieved when all the $a_i$ are equal to $\dim(U)/k$.
Thus,
\begin{align*}
 \Pr_{y,z} [ y \otimes z \in U] &\leq \frac{1}{2^k}\left( 1 +  \sum_{i \in [k]}  2^{i-1} f_{d-1, k}(u/k) \right)\\
  &= \frac{1}{2^k}\left(  1 +   (2^k-1) f_{d-1,k}(u/k) \right)\\
  &= \left( \frac{1}{2^k} + (1 - \frac{1}{2^k}) f_{d-1,k}(u/k) \right) \\
  &= f_{d,k}(u).
\end{align*}
This completes the induction step.
\end{proof}
\begin{theorem}\label{thm:max}
Let $k$ be a positive integer, and let $\tilde{u} \in [0, k^2]$ be a real number.
Suppose $b_1, b_2, \ldots, b_k$ are real numbers satisfying the following constraints.
\begin{align}
k \geq b_1 \geq b_2 \ldots \geq b_k \geq 0, \label{ordereq}\\
\sum_{i=1}^k b_i = \tilde{u}. \label{sumeq}
\end{align}
Then,
$$ \sum_{i=1}^{k} 2^{i-1} 2^{b_i} \leq \sum_{i=1}^{k} 2^{i-1} 2^{\tilde{u}/k} = (2^k-1) 2^{\tilde{u}/k}.$$
\end{theorem}
We prove~\cref{thm:max} in the appendix and now use the previous lemma to prove a corollary about
the dimension of the span of several random rank 1 tensors.

\begin{corollary}
\label{cor:onebyone}
Let $d, k, t$ be integers.
For each $i \in [t]$ and $j \in [d]$, pick $x^{(i,j)} \in \F_2^k$ uniformly at random.
For $i \in [t]$, let $T_i$ be the rank-$1$ tensor $\otimes_{j=1}^d x^{(i,j)}$.
Then, for every $ 0 \le r \le t$,
$$\Pr[\mathrm{dim}(\mathrm{span}(\{T_1, \ldots, T_t\})) = r] \le {t \choose r} \left( \frac{d + 2^{t/k^{d-1}}}{2^k} \right)^{t-r}.$$ 
\end{corollary}
\begin{proof}
Let us reveal $T_1, \ldots, T_t$ one at a time.
For $0 \leq i \leq t$, let $V_i = \mathrm{span}(\{T_1, \ldots, T_{i-1}, T_i \})$.
Thus we have $0 = \dim(V_0) \leq \dim(V_1) \leq \ldots \dim(V_t)$.
We want to estimate the probability that $\dim(V_t) = r$.
Let $E_i$ denote the event that $T_i \in V_{i-1}$. For $I \subseteq [t]$, let $E_I$ denote the event $\bigcap_{i \in I} E_i$. In terms of these events, we can bound $\Pr[ \dim (V_t) = r] $ as follows.
\begin{align*} \Pr[ \dim (V_t) = r]  &\leq  \Pr[ \exists I \subseteq [t], |I| = t-r \mbox{ such that } E_I \mbox{ occurs} ]\\
& \leq \sum_{I \subseteq [t], |I| = t-r } \Pr[E_I].
\end{align*}
We conclude the proof by bounding $\Pr[E_I]$. Fix $I \subseteq [t]$ with $|I| = t-r$.
Let $I = \{i_1, \ldots, i_{t-r} \}$ with $i_1 < i_2 < \ldots < i_{t-r}$.
$$\Pr[E_I] = \prod_{j=1}^{t-r} \Pr[ E_{i_{j}}| \bigcap_{\ell < j} E_{i_{\ell}} ].$$
\cref{lem:randrankone} implies the following.
$$ \Pr[E_{i} | T_1, \ldots, T_{i-1} ] \leq \frac{d + 2^{\dim(V_{i-1})/k^{d-1}}}{2^k}.$$
For any given $j \in [t-r]$,
the events $E_{i_1}, \ldots, E_{i_{j-1}}$ are all determined by 
$T_1, \ldots, T_{i_j-1}$ (since $E_{i_\ell}$ depends on $T_1, \ldots, T_{i_\ell}$, and
$i_{j-1} \leq i_j - 1$).
Thus, for each $j \in [t-r]$, we have,
$$ \Pr[E_{i_j} |  \bigcap_{\ell < j} E_{i_{\ell}}   ] \leq \frac{d + 2^{t/k^{d-1}}}{2^k}.$$
Here we used the fact that $\dim(V_{i_j-1}) \leq t$.
Using this in our previous bound, we conclude that
$$\Pr[E_I] \leq \left(\frac{d + 2^{t/k^{d-1}}}{2^k}\right)^{t-r},$$
and thus,
\begin{equation*}\Pr[\dim(V_t) =r ] \leq {t \choose r} \cdot
  \left(\frac{d + 2^{t/k^{d-1}}}{2^k}\right)^{t-r}.\qedhere
\end{equation*}
\end{proof}

\subsection{Proof of \cref{thm:tensorsumbound}}
We now use \cref{cor:onebyone} to prove \cref{thm:tensorsumbound}.
\begin{proof}[Proof of \cref{thm:tensorsumbound}]
The equation
\begin{equation}
\label{eq:tensorzero}
 \sum_{i = 1}^t \bigotimes_{j =1}^d x^{(i,j)} = 0
\end{equation}
implies that
\begin{align}
\label{eq:tensorflatzero}
\forall \ell \in [k],\ \sum_{i = 1}^t x^{(i,1)}_\ell \cdot \bigotimes_{j =2}^d x^{(i,j)} = 0.
\end{align}
Let $T_i$ denote $\bigotimes_{j=2}^d x^{(i,j)}$ for $i \in [t]$ and $\mathcal{T} = \mathrm{span}(\{T_1,\ldots, T_t\})$. Then we have,
\begin{align}
\label{eq:tensordim}
\Pr[ \{x^{(i,j)}\}_{i \in [t], j \in [d]}\mbox{ satisfy~\eqref{eq:tensorzero}}] &\le \Pr[ \{x^{(i,j)}\}_{i \in [t], j \in [d]}\mbox{ satisfy~\eqref{eq:tensorflatzero}}]\nonumber\\
&= \sum_{r = 0}^t \Pr\left[ \{x^{(i,j)}\}_{i \in [t], j \in [d]}\mbox{ satisfy~\eqref{eq:tensorflatzero}} \big| \mathrm{dim}(\mathcal{T}) = r \right] \Pr\left[ \mathrm{dim}(\mathcal{T}) = r \right]\nonumber\\
&= \sum_{r=0}^t \left(\prod_{\ell \in [k]} \Pr\left[\sum_{i = 1}^t x^{(i,1)}_\ell\cdot T_i = 0\big| \mathrm{dim}(\mathcal{T}) = r \right]\right)\cdot \Pr\left[ \mathrm{dim}(\mathcal{T}) = r \right] \nonumber\\
&\le \sum_{r=0}^t \left(\frac{1}{2^r}\right)^k \cdot \Pr\left[ \mathrm{dim}(\mathcal{T}) = r \right].
\end{align}
Here, the equality in the third step follows from the fact that $\{x^{(i,1)}_\ell\}_{i \in [t], \ell \in [k]}$ are independently and uniformly distributed in $\F_2$.\\
By the given distribution of $T_1, \ldots, T_t$ in $(\F_2^k)^{\otimes (d-1)}$, \cref{cor:onebyone} says that
$$\Pr\left[ \mathrm{dim}(\mathcal{T}) = r \right] \le {t \choose r} \left(\frac{d-1 + 2^{\frac{t}{k^{d-2}}}}{2^k}\right)^{t-r}.$$
Plugging this bound back into \eqref{eq:tensordim} gives
\begin{align*}
\Pr[ \{x^{(i,j)}\}_{i \in [t], j \in [d]}\mbox{ satisfy~\eqref{eq:tensorzero}}] &\le \sum_{r=0}^t {t \choose r} \frac{1}{2^{rk}} \left(\frac{d-1 + 2^{\frac{t}{k^{d-2}}}}{2^k}\right)^{t-r} \nonumber \\
&\le \sum_{r=0}^t {t \choose r} \left(\frac{1}{2^k}\right)^r \left(\frac{d-1 + 2^{\frac{t}{k^{d-2}}}}{2^k}\right)^{t-r}\nonumber\\
&= \left(\frac{1}{2^k} + \frac{d-1 + 2^{\frac{t}{k^{d-2}}}}{2^k} \right)^t \nonumber\\
&\le \left( \frac{d + 2^{\frac{t}{k^{d-2}}}}{2^k}\right)^t .\\
\end{align*}
Now, since $d <  2^{\epsilon k/5}$ and $t < \epsilon k^{d-1}/5$, 
we have
$$ d + 2^{\frac{t}{k^{d-2}}} < 2 \cdot  2^{\epsilon k/5} < 2^{\epsilon k /2},$$
we conclude that 
$$\Pr[\sum_{i = 1}^t \bigotimes_{j =1}^d x^{(i,j)} = 0] <  2^{-(1-\epsilon/2) kt}.$$
This completes the proof. 
\end{proof} 
\subsection{Explicit $d$-linear forms with small correlations with $d-1$-linear forms}
In this section, we dwell a bit on the question of constructing explicit $d$-linear forms which have small correlation with lower degree multilinear polynomials. In particular, we present an explicit $d$-linear form that has exponentially small correlation with any lower degree multilinear form. Define $f:(\F_2^k)^d \to \F_2$ as 
$$f(x_1,...,x_d)= \langle x_1\cdot x_2\cdots x_{d-1}, x_d\rangle,$$
where $\cdot $ denotes multiplication over the bigger field
$\F_{2^k}$. It is easy to see that $f$ is a $d$-linear form. 

Ideally, we would like to show that the map $f$ defined above has small correlation with \emph{any} polynomial of degree at most $d-1$. But, we do not know how to show this. In the rest of this section, we show that $f$ has small correlation with any polynomial of degree $d-1$ which respects the partition of the inputs to $f$. We now prove the following lemma.
\begin{lemma}
The function $f=\langle x_1\cdot x_2\cdots x_{d-1}, x_d\rangle$ has correlation at most $(d-1)2^{-k}$ with any degree $\leq d-1$ multilinear form. 
\end{lemma}
\begin{proof}
Let $g$ be a $(d-1)$-linear form over $(\F_2^k)^d$. A similar proof as below works for any $d'$-linear form $g'$ for $d'< d-1$ also. We want to understand
$$
\corr(f,g)= \bias(f-g).
$$
Since $g$ is a $(d-1)$-form, it is of the form
$g(x_1,...,x_d)=\sum_{i=1}^d g_i(x_{[d]\setminus \{i\}})$. Since 
$g_i$ is a $(d-1)$-linear form in the variables
$x_{[d]\setminus\{i\}}$, for each $i \in [d-1]$ there exists an
$\F_2^k$-valued linear form $v_i = v_i(x_{[d-1]\setminus \{i\}})$ such
that $g_i(x_{[d] \setminus \{i\}}) = \langle v_i(x_{[d-1]\setminus
  \{i\}}), x_d\rangle = \langle v_i,x_d\rangle$. In particular
\begin{equation}\label{rewritingthebias}
\corr(f,g)= \bias(f-g)= \bias\left((f-\sum_{i=1}^{d-1} g_i) -
  g_d\right) =  \Pr_{x_1,...x_{d-1}} \left[x_1\cdot x_2\cdots x_{d-1}-\sum_{i=1}^{d-1} v_i=\overline{0}\right].
\end{equation}
This is because $g_d$ does not depend on $x_d$ and for any fixing of
$x_1,...,x_{d-1}\in \F_2^k$, $f-g$ is an affine form in the variable
$x_d$ that is biased if $x_1\cdot x_2\cdots x_{d-1}-\sum_{i=1}^{d-1}
v_i=\overline{0}$ (in which case the bias is 1) and is otherwise an
unbiased function. We will prove 
\begin{equation}\label{repeatedfact}
\Pr_{x_1,...x_{d-1}\in \F_2^k} \left[x_1\cdot x_2\cdot \cdots
  x_{d-1}-\sum_{i=1}^{d-1} g_i=\overline{0}\right]\leq
\Pr_{x_1,...,x_{d-1}\in \F_2^k}[x_1\cdot x_2\cdot \cdots x_{d-1}=\overline{0}],
\end{equation}
by repeatedly applying the following fact. 
\begin{fact}\label{linearmapfact}
Let $h:\F_2^k\to \F_2^k$ be a linear map. Then for every $\overline{a}\in \F_2^k$, 
$$
\Pr_{x\in \F_2^k} [h(x)=\overline{a}] \leq \Pr_{x\in \F_2^k} [h(x)=\overline{0}]. 
$$	
\end{fact}

Note that applying this fact we have
\begin{align*}
\Pr_{x_1,...x_{d-1}} \left[x_1\cdot x_2\cdots x_{d-1}-\sum_{i=1}^{d-1}
  v_i=\overline{0}\right]&= \Pr_{x_1,...x_{d-1}} \left[x_1\cdot
                           x_2\cdots x_{d-1}-\sum_{i=1}^{d-2}
                           v_i=v_{d-1}\right]\\ &\leq
                                                  \Pr_{x_1,...x_{d-1}}
                                                  \left[x_1\cdot x_2\cdots x_{d-1}-\sum_{i=1}^{d-2} v_i = \overline{0}\right],
\end{align*}
since for every fixing of $x_1,...,x_{d-2}\in \F_2^k$, $x_1\cdot
x_2\cdots x_{d-1}-\sum_{i=1}^{d-2} v_i$ is a $\F_2^k$-linear form over
$x_{d-1}$. Applying \cref{linearmapfact} similarly for coordinates
$i=1,...,d-2$, we get \cref{repeatedfact}. Finally, by a simple union
bound we can bound $\Pr_{x_1,...,x_{d-1}\in \F_2^k}[x_1\cdot x_2\cdots
x_{d-1}=\overline{0}] = 1 - (1-2^{-k})^{d-1}\leq (d-1)\cdot 2^{-k}$. Combining this with \cref{repeatedfact} and \cref{rewritingthebias} finishes our proof. 
\end{proof}

\section{High-rank tensors from unbiased polynomials}

It is well-known that the bias of a bilinear form corresponding to a matrix $M \in
\F_2^{k \times k}$ is tightly related to its rank $\rank(M)$ (more
precisely, $\bias(M) = 2^{-\rank(M)}$). In this section, we explore a similar
connection for higher dimensional tensors. We then use this to
(re)prove some existing tensor rank lower bounds (e.g., for the trace tensor
and the matrix multiplication tensor)

\subsection{Small Bias implies large tensor rank}

We begin with the main theorem of this section which shows tensors with
small bias have large rank. 
\begin{theorem}[Small bias implies large rank]\label{thm:lowrankbias-gen}
Let $P\in \F_2^{k\times k\cdots \times k}$ be any $d$-dimensional tensor of rank $\leq t$. Then  $$\bias(P)\geq \left(1-\frac{2}{2^d} \right)^t.$$  	
\end{theorem}
An important ingredient of our proof will be the following lemma.
\begin{lemma}\label{lem:low rank bias}
Let $d$ be a natural number.  Let $M_1, M_2, \ldots, M_t \in \F_2^{k \times k \cdots \times k}$ be $d$-dimensional tensors of rank at most $1$. Then,
\begin{equation}
\Pr_{x_1, x_2, \ldots, x_d \in \F_2^k}\left[ \forall i\in [t], M_i (x_1, x_2, \ldots, x_d) =0\right] \geq \left(1-\frac{1}{2^d} \right)^t. 
\end{equation}
\end{lemma}
\begin{proof}
Our proof is by induction on $d$. 
\paragraph*{Base Case. }The base case when $d=1$ trivially follows
  since if there are $t$ linear forms $u_1, u_2, \ldots, u_t$ over
  $\F_2$, then the  maximum number $r$ of independent linear forms among them
  is at most $t$. We hence have, 
\begin{equation}
\Pr_{x\in \F_2^k}\left[ \forall i\in [t], u_i(x) =0\right] =
\left(1/2\right)^r \geq \left(1/2\right)^t \, . 
\end{equation}
\paragraph*{Induction Step.}
Before proving the general inductive step
  from $d-1$ to $d$, we first show the $d=2$ case as a warm up as it
  illustrates the main idea and then do the general case. 

For this case, we have $k \times k$ matrices $M_1, M_2, \ldots, M_t$ of rank one over $\F_2$, and the goal is to show that 
\begin{equation}
\Pr_{y,z\in \F_2^k}[ \forall i\in [t], \langle y, M_i z\rangle =0] \geq (3/4)^t \, . 
\end{equation}
The proof involves several steps of manipulation of the probability of interest. For a set $S\subseteq [t]$, denote by $M_S := \sum_{i\in S} M_i$. 
\begin{align*}
\Pr_{y,z\in \F_2^k} \left[ \forall i\in [t], \ip{y, M_i z} =0\right] &=
\E_{y,z\in \F_2^k} \left[ \prod_{i=1}^t \left(\frac{1+(-1)^{\ip{y, M_iz}}}{2}\right)\right]\\ &=
\E_{y,z\in \F_2^k}\left[ \frac{1}{2^t} \cdot \sum_{S\subseteq [t]} (-1)^{\ip{y, M_Sz}} \right] \\ &=
\E_{y,z\in \F_2^k} \left[\E_{S\subseteq [t]} \left[ (-1)^{\ip{y,M_Sz}}\right] \right]\\ &= 
\E_{S\subseteq [t]} \left[ \E_{y,z\in \F_2^k}\left[ (-1)^{\ip{y,M_Sz}}\right]\right]\\ &=
\E_{S\subseteq [t]}\left[ \E_{z\in \F_2^k} \left[ 1_{M_Sz=\overline{0}}\right]\right]
\\ &=
\E_{S\subseteq [t]}\left[ \Pr_{z\in \F_2^k} \left[M_Sz=\overline{0}\right]\right]\\ &=
\E_{S\subseteq [t]}\left[2^{-\rank(M_S)}\right]\\ &\geq \E_{S\subseteq [t]}\left[ 2^{-|S|}\right]=
\frac{1}{2^t} \cdot \left(1+\frac{1}{2}\right)^t = \left(\frac{3}{4}\right)^t.
\end{align*}
Now, for the general inductive step, we assume that the lemma is true up to dimension $d-1$, and prove it for $d$ dimensions.  For every $i \in [t]$, we denote by $u_i$ as the linear form in $\F_2^k$ and $M_{i}'$ as the $d-1$ dimensional tensor of rank $1$ in $\F_2^{k \times k \times k \cdots \times k}$ such that 
\[
M_i(x_1, x_2, \ldots, x_d) = u_i(x_1) \cdot M_i'(x_2, x_3, \ldots, x_d) \, .
\]
And, once again, for every $S\subseteq [t]$, $M_S$ denotes the tensor $\sum_{j \in S} M_j$, which has rank at most $|S|$. We proceed via a sequence of inequalities as in the case of $d = 2$ above.
\begin{align*}
\Pr_{x_1,x_2, \ldots, x_d\in \F_2^k} \left[ \forall i\in [t],  M_i (x_1, x_2, \ldots, x_d) =0\right] &=
\E_{x_1,x_2, \ldots, x_d\in \F_2^k} \left[ \prod_{i=1}^t \left(\frac{1+(-1)^{M_i (x_1, x_2, \ldots, x_d)}}{2}\right)\right]\\ &=
\E_{x_1,x_2, \ldots, x_d\in \F_2^k}\left[ \frac{1}{2^t} \cdot \sum_{S\subseteq [t]} (-1)^{M_S (x_1, x_2, \ldots, x_d)} \right] \\ &=
\E_{x_1,x_2, \ldots, x_d\in \F_2^k} \left[\E_{S\subseteq [t]} \left[ (-1)^{M_S (x_1, x_2, \ldots, x_d)}\right] \right]\\ &= 
\E_{S\subseteq [t]} \left[ \E_{x_1,x_2, \ldots, x_d\in \F_2^k}\left[ (-1)^{M_S (x_1, x_2, \ldots, x_d)}\right]\right]\,  .
\end{align*}
Now, observe that for every $S\subseteq [t]$, \[
\E_{x_1,x_2, \ldots, x_d\in \F_2^k}\left[ (-1)^{M_S (x_1, x_2, \ldots, x_d)}\right]\geq \Pr_{x_2, x_3, \ldots, x_d}\left[\forall j \in S, M_j'(x_2, x_3, \ldots, x_d) = 0 \right]\, .
\]
Moreover, from the induction hypothesis, we get that for all $S\subseteq [t]$,
\[\Pr_{x_2, x_3, \ldots, x_d}\left[\forall j \in S, M_j'(x_2, x_3, \ldots, x_d) = 0 \right] \geq \left(1-\frac{1}{2^{d-1}}\right)^{|S|}\, .
\]
Plugging this back in the calculations, we get 
\begin{align*}
\Pr_{x_1,x_2, \ldots, x_d\in \F_2^k} \left[ \forall i\in [t], M_i (x_1, x_2, \ldots, x_d) =0\right] &\geq 
\E_{S\subseteq [t]}\left[\left(1-\frac{1}{2^{d-1}}\right)^{|S|}\right]\\ 
&\geq \frac{1}{2^t} \cdot \left(1+1-\frac{1}{2^{d-1}}\right)^t =
  \left(1-\frac{1}{2^d}\right)^t. \qedhere
\end{align*}
\end{proof}
We now complete the proof of~\cref{thm:lowrankbias-gen}.
\begin{proof}[Proof of \cref{thm:lowrankbias-gen}]
Since $P$ has rank $\leq t$, then there is a collection of linear forms $u_1, u_2, \ldots, u_t$ and tensors $M_1, M_2, \ldots, M_t$ of rank at most $1$ in $d-1$ dimensions such that 
$$
P(X_1, X_2, \ldots, X_d)= \sum_{i=1}^t u_i(X_1) \cdot M_i(X_2, X_3, \ldots, X_d) \, . 
$$
Now, observe that  
\begin{align*}
\bias(P)&= \abs{\E_{x_1, x_2, \ldots, x_d \in \F_2^k}\left[ (-1)^{P(x_1, x_2, \ldots, x_d)}  \right]} \\
&=  \Pr_{x_2, x_3, \ldots, x_d\in \F_2^k} \left[P(X_1, x_2, x_3, \ldots, x_d)\equiv 0\right] \\ 
&\geq    \Pr_{x_2, x_3, \ldots, x_d\in \F_2^k} \left[ \forall i\in [t], \;M_i(x_2, x_3, \ldots, x_d)= 0 \right] \\
&\geq \left(1-\frac{1}{2^{d-1}}\right)^t\qquad \text{[By \cref{lem:low rank bias}]}\, . \qedhere
\end{align*}
\end{proof}

We now accompany the above theorem with an almost matching upper bound
on the bias of random high rank tensors. It is known that a random
high rank tensor has low bias. The following lemma gives a precise
quantitative version of this observation (the idea for the proof was suggested to us by Shubhangi Saraf).

\begin{lemma}\label{lem:bias ub for random tensor}
For $i \in [t]$ and $j \in [d]$, let $u_{i,j} \in \F_2^k$ be a uniformly random vector. Consider the random rank-$t$ $d$-linear form $p: (\F_2^k)^d \to \F_2$ given by
$$p (x_1,x_2, \ldots, x_d)= \sum_{i=1}^t \prod_{j = 1}^d \ip{x_j,u_{i,j}}.$$
Then 
$$
\E [\bias(p)] \leq  d\cdot 2^{-k} +  \left(1-\frac{2}{2^d}\right)^t 
$$
\end{lemma}
\begin{proof}
We have
\begin{align*}
\E_{p}[\bias(p)]&= 
\E_p \E_{x_{1}, x_{2}, \ldots, x_{d}\in \F_2^k} \left[(-1)^{\sum_{i=1}^t \prod_{j = 1}^d \ip{x_j,u_{i,j}}}\right] \\
&= \E_{x_{1}, x_{2}, \ldots, x_{d}\in \F_2^k} \E_p
  \left[(-1)^{\sum_{i=1}^t \prod_{j = 1}^d \ip{x_j,u_{i,j}}}\right]\\
& = \Pr_{x_1,\dots,x_d}\left[ \exists i, x_i = \overline{0}\right]
  + \Pr_{x_1,\dots,x_d}\left[ \forall i, x_i \neq
  \overline{0}\right] \cdot \E_{{x_{1}, x_{2}, \ldots, x_{d}\in
  \F_2^k\backslash \{\overline{0}\}}}\left[\prod_{i=1}^t
  \left(\E_{u_{i,1},u_{i,2},\ldots, u_{i,d} }(-1)^{ \prod_{j
  = 1}^d \ip{x_j,u_{i, j}}}\right)\right]\\
& = 1 - \left(1-\frac1{2^k}\right)^d +\left(1-\frac{1}{2^k}\right)^d\cdot \E_{{x_{1}, x_{2}, \ldots, x_{d}\in
  \F_2^k\backslash
  \{\overline{0}\}}}\left[\prod_{i=1}^t\left(\Pr_{u_{i,1},\ldots,u_{i,d-1}}\left[\exists j \in [d-1], \langle x_j, u_{i,j}\rangle  =
  0\right]\right)\right]\\
&= 1 - \left(1-\frac1{2^k}\right)^d +\left(1-\frac{1}{2^k}\right)^d\cdot\left(1-\frac1{2^{d-1}}\right)^t\\
&\leq   d\cdot 2^{-k} +  \left(1-\frac{2}{2^d}\right)^t \, .\qedhere
\end{align*}
\end{proof}

The following special cases of~\cref{thm:lowrankbias-gen}, for $d=2$ and $d = 3$ will be useful for us, on our way to proving lower bounds on the rank of three dimensional tensors. 

\begin{corollary}\label{thm:lowrankbias-2d}
Let $P\in \F_2^{k\times k}$ be a matrix of rank $\leq t\leq k$. Then,  $\bias(P)\geq 2^{-t}$.  	
\end{corollary}
\begin{corollary}\label{thm:lowrankbias}
Let $P\in \F_2^{k\times k\times k}$ be a $3$-dimensional tensor of rank $\leq t$. Then,  $\bias(P)\geq \left(\frac{3}{4}\right)^t$.  	
\end{corollary}
In the subsequent two sections, we will observe that some well-known
explicit tensors in three dimensions have very low bias, and then use
the above corollaries to conclude that these tensors have large rank. 

\subsection{A 3.52k Tensor Rank Lower Bound for $\tr(XYZ)$}\label{sec:trace}

In this section, we use the bias-vs-tensor-rank connection explored in
the previous section to construct explicit 3-dimensional tensors with large tensor
rank. \cref{thm:lowrankbias} suggests the following natural
approach to construct tensors of large rank: find a 3-linear form with
as small a bias as possible. What is the least bias of a 3-linear
form? Let $P(X,Y,Z) = \sum_{i=1}^k \langle Y,
M_iZ\rangle X_i$ be an arbitrary 3-linear form. Clearly, $\bias(P)
\ge \Pr_{y,z}[\forall i \in [k], \; \langle y, M_iz \rangle =0]
\geq \Pr_{y,z}[ y = \overline{0} \text{ or } z = \overline{0}] =
2/2^{k} -1/2^{2k}$. The $\tr(XYZ)$ is a function with bias exactly
$2/2^{k}-1/2^{2k}$ (see \cref{lem:bias of trace}). In the rest of this
section, we prove an upper bound on the bias of this function. 
To this end, we first show that the bias of
$Tr(X, Y, Z)$ is small. This will immediately via
\cref{thm:lowrankbias} give a very simple proof that $\tr(XYZ)$ tensor
has rank at least $2.409k$.  We remark that a much stronger rank
lower-bound of $3.52k$ is known due to Chudnovsky and
Chudnovsky~\cite{ChudnovskyC1988,ShparlinskiTV1992} and indeed we do a more careful
analysis of our ideas to get a new proof of the $3.52k$ lower bound (here too
the only property of $Tr$ that is used is that it is of very low bias). 
\begin{lemma}\label{lem:bias of trace}
$$
\bias(Tr(X,Y,Z)) =2\cdot 2^{-k} - 2^{-2k}. 
$$
\end{lemma}
\begin{proof}
The trace function satisfies the simple property that for every
non-zero $\alpha \in \F_{2^k}$, the linear function $\tr(\alpha X)$ is
unbiased. Hence, 
\begin{equation*}
\bias(Tr(X,Y,Z)) = \Pr_{x,y \in \F_2^k} \left[ x\cdot y = 0
                   \right] = 2\cdot 2^{-k} - 2^{-2k}\,.\qedhere
\end{equation*}
\end{proof}

The above lemma coupled with \cref{thm:lowrankbias} immediately gives
the following lower bound on tensor rank of $Tr(X,Y,Z)$. 
\begin{corollary}\label{cor : weaker tensor rank lower bound}
$\rank(Tr(X,Y,Z)) \geq (\log_{4/3}2)\cdot k \geq 2.409k. $
\end{corollary}

We now strengthen this bound to show a $3.52 k$ lower bound on the
rank of $Tr(X, Y, Z)$. As we alluded to in earlier discussion, this
matches the best known lower bound on the tensor rank of \emph{any}
explicit tensor in three dimensions. The proof follows from a more
careful use of the ideas already present in the proof of~\cref{cor :
  weaker tensor rank lower bound}. We will need the following
well-known rate-distance MRRW tradeoff for linear codes.  
\begin{theorem}[\cite{McelieceRRW1977}]\label{thm:MRRW}
Let $S$ be a subspace of dimension at least $k$ of $\F_2^t$, such that
every non-zero vector in $S$ has weight at least $k$. Then, $t \geq
3.52k$.\footnote{The MRRW bound for binary codes states that any family
  of codes with fractional distance $\delta$ satisfies $R(\delta) \leq
  h_2\left(\frac12-\sqrt{\delta(1-\delta)}\right)$ where $h_2(x) =
  x\log_2(1/x) + (1-x)\log_2(1/1-x)$ is the binary entropy function. The above mentioned
  bound can be obtained from this (see \cite{BrownD1980} for details).}	
\end{theorem}
 
\begin{theorem}\label{thm : 3.5 lb on tensor rank}
The rank of the tensor $Tr(X, Y, Z)$ is at least $3.52k$. 
\end{theorem}
\begin{proof} Let the tensor rank of $Tr(X,Y,Z)$ be $t$. Then there
  exists $t$ vectors $a_1, a_2,\ldots, a_t \in \F_2^k$ and $t$
  rank-$1$ matrices $M_1, M_2, \ldots, M_t$ such that 
\begin{equation}\label{eq:t-tensor}
Tr(X, Y, Z) = \sum_{i=1}^t \ip{a_i, X}\cdot \ip{Y, M_iZ} \, .
\end{equation}
Let $A$ be the $k \times t$ matrix such that for every $i \in [t]$, the $i^{th}$ column of $A$ equals $a_i$. Let $K$ be the kernel of $A$. Clearly, $\dim(K) \geq t-k$.  In fact, $\dim(K) = t-k$. To see this, observe that if $\dim(K) \geq t-k + 1$, then by the rank-nullity theorem, $\rank(A) \leq k-1$. Thus, there is a non-zero $x \in F_2^k$ denoted by $x_0$ such that for every $i \in [t]$, $\ip{a_i, x_0} = 0$. Thus, $Tr(x_0, Y, Z) \equiv 0$ for a non-zero $x_0$, which is a contradiction. 

From proof of~\cref{thm:lowrankbias}, we know that 
\[
\bias(Tr(X, Y, Z)) = {\Pr_{y, z \in \F_2^k} [Tr(X, y, z) = 0]} \, .
\]
So far we were proving a lower bound on ${\Pr_{y, z \in \F_2^k} [Tr(X, y, z) = 0]}$ by proving a lower bound on $\Pr_{y, z \in \F_2^k}\left[ \forall i \in [t], \ip{y, M_iz} = 0\right]$. Clearly, this seems to be somewhat lossy since even for a choice of $y $ and $z$ in $\F_2^k$ such that $\ip{y, M_iz} \neq 0$ for some $i \in [t]$, it is conceivable that $Tr(X, y, z)$ is identically zero. For this proof, we try to be a bit more careful about this. Note that for every $u \in K \subset \F_2^t$, 
\[
\sum_{i = 1}^t u_i\cdot \ip{a_i, X} \equiv 0 \, .
\]
Thus, we have,
\begin{align*}
\Pr_{y, z \in \F_2^k} [Tr(X, y, z) = 0] &=\sum_{u \in K} \Pr_{y, z \in \F_2^k} \left[\forall i \in [t], \ip{y, M_iz} = u_i \right] \\
&= \sum_{u \in K} \E_{y, z} \left[\prod_{i \in [t]} \left( \frac{1 + (-1)^{\ip{y, M_iz} + u_i}}{2}\right) \right] \\
&=\sum_{u \in K} \E_{y, z} \left[\E_{S\subseteq [t]} (-1)^{\ip{y, M_S z}} \cdot (-1)^{\ip{u, 1_{S}}} \right] \, .\\
\end{align*}
Here, for every $S \subseteq [t]$, $1_S$ is the characteristic vector of $S$ in $t$ dimensions, and $M_S = \sum_{i \in S} M_i$. Simplifying further, we get, 
\begin{align*}
\Pr_{y, z \in \F_2^k} [Tr(X, y, z) = 0] &=\E_{S\subseteq [t]} \left[ \left( \E_{y, z}  (-1)^{\ip{y, M_S z}} \right) \cdot \left(\sum_{u \in K} (-1)^{\ip{u, 1_{S}}}\right) \right] \, .\\
\end{align*}
Now, we observe that the term $\left(\sum_{u \in K} (-1)^{\ip{u,
      1_{S}}}\right) = \abs{K}$ if and only if $1_S \in K^{\perp}$,
otherwise it equals zero. Also, from~\cref{thm:lowrankbias-2d}, we
know that $\left( \E_{y, z}  (-1)^{\ip{y, M_S z}} \right)=2^{-\rank{M_S}}$ is at at least $\max\{2^{-k}, 2^{-|S|}\}$. Plugging these into the inequality above, we have the following inequality. 
\begin{align*}
\Pr_{y, z \in \F_2^k} [Tr(X, y, z) = 0] &\geq \frac{\abs{K}}{2^t} \cdot \sum_{v \in K^{\perp}} \max\{2^{-k}, 2^{-|v|}\}  &&[\text{Here, } |v| \text{ is the Hamming weight of } v] \\
&\geq \E_{v \in K^{\perp}} \max\{2^{-k}, 2^{-|v|}\} &&[\text{ Since } \abs{K}\cdot \abs{K^{\perp}} = 2^t]
\end{align*}
Recall that the dimension of $K^{\perp}$ equals $k$. Now, 
\[
\E_{v \in K^{\perp}} \max\{2^{-k}, 2^{-|v|}\} = 2^{-k} + \E_{v \in K^{\perp}\setminus \{0^k\}} \max\{2^{-k}, 2^{-|v|}\}\, .
\] 
From~\cref{lem:bias of trace}, we know that the bias of $Tr(X, Y, Z)$ is at most $2\cdot 2^{-k} - 2^{-2k}$. Thus, it must be the case that $\E_{v \in K^{\perp}\setminus \{0^k\}} \max\{2^{-k}, 2^{-|v|}\} \leq (1-2^{-k})\cdot 2^{-k}$. But this is possible only if all the vectors in $K^{\perp}\setminus \{0^k\}$ have weight at least $k$. In this case, the space $K^{\perp}$ is a linear subspace of $\F_2^t$ of dimension $k$ such that every non-zero vector in it has Hamming weight at least $k$. From~\cref{thm:MRRW}, we get that $t \geq 3.52k$. This completes the proof.
\end{proof}

\subsection{Lower Bound on the Rank of Matrix Multiplication Tensor}
In this section, we obtain a lower bound on the rank of the matrix multiplication tensor by proving an upper bound on its bias. Even though better bounds are known for this tensor, 
our proof is a fairly straightforward application of our techniques, and we believe this is instructive. 

Our main technical observation in this section is the following lemma which gives an upper bound on the bias of $M_{n}(\overline{X}, \overline{Y}, \overline{Z})$ as each of the variables take values in $\F_2$.

\begin{lemma}\label{lem:bias of MM}
The bias of $M_n(\overline{X}, \overline{Y}, \overline{Z})$ is at most $n \cdot 2^{-\frac{3n^2}{4}}$.
\end{lemma}
Before proceeding with the proof, we note that~\autoref{lem:bias of MM} and~\cref{thm:lowrankbias} immediately imply a non-trivial lower bound on the tensor rank of $M_n$. 
\begin{theorem}\label{thm:trank of MM}
The tensor rank of $M_n$ is at least $\frac{3n^2}{4\log_2(4/3)}\geq 1.8n^2$.
\end{theorem}
We now prove~\cref{lem:bias of MM}. 
\begin{proof}[Proof of~\cref{lem:bias of MM}]
We observe that for any two fixed matrices $x,y$, the 3-linear form $M_n$ reduces
to a linear form in $z$ which is non-zero iff the product of the two
matrices $x$ and $y$ is non-zero. Furthermore, given a matrix $y$, the probability (over $x$) that the product matrix $x\cdot
y$ is zero is exactly $2^{-n\cdot \rank(y)}$. Combining these observations, we have
\begin{align*}
\bias(M_n) & = \Pr_{x, y}\left[ x\cdot y = 0_{n\times n}\right]\\
&=\E_y \left[2^{-n\cdot \rank(y)} \right] \\
&=\sum_{r = 0}^n \Pr_{y}\left[ \rank(y) = r\right] \cdot 2^{-nr}\;.
\end{align*}
To complete the proof, we rely on the following claim, whose proof we defer to the end of this section.
\begin{claim}\label{clm:rank distr}
For every $r \in \set{0, 1, \ldots, n}$, the following inequality is true.
\[
\Pr_{y}\left[\rank(y) = r \right] \leq 2^{-(n-r)^2} \, .
\]
\end{claim}
From the claim above, we get 
\begin{align*}
\bias(M_n) &\leq \sum_{r = 0}^n 2^{-(n-r)^2 - nr} \\
&\leq \sum_{r = 0}^n 2^{-n^2 - r^2 + nr} \\
&\leq 2^{-n^2}\sum_{r = 0}^n 2^{ r(n-r)} \\
&\leq 2^{-n^2} n\cdot 2^{n^2/4} \\
&\leq n\cdot 2^{-3n^2/4}\; . \qedhere
\end{align*}
\end{proof}
For completeness, we now provide a proof of~\cref{clm:rank distr}. We remark
that the following tighter bound is known (see
\cite[Theorem~3.2.1]{Kolchin}).
\begin{align*}
\Pr_y\left[\rank(y) = r\right]
& = 2^{-(n-r)^2}\cdot
\prod_{i=n-r+1}^n\left(1-\frac1{2^i}\right) \cdot \left(\sum_{0 \leq
    i_1\leq \ldots i_{n-r}\leq r} \frac1{2^{i_1 + \ldots +
  i_{n-r}}}\right)\\
& \leq 2^{-(n-r)^2}\cdot
\prod_{i=n-r+1}^n\left(1-\frac1{2^i}\right) \cdot 
\prod_{i=1}^{n-r}\left(1-\frac1{2^i}\right)^{-1}\;.
\end{align*}
However, the weaker bound given in the claim suffices for our purposes.
\begin{proof}[Proof of~\cref{clm:rank distr}]
The goal is to upper bound the probability that a uniformly random
$n\times n$ matrix $y$ over $\F_2$ has rank equal to $r$. This
probability is upper bounded by the probability that the rows of $y$
are contained within a subspace of dimension $r$ of $\F_2^n$. For any
fixed subspace $S$ of dimension equal to $r$, this event happens with
a probability equal to $2^{-n(n-r)}$. The number of subspaces of
$\F_2^n$ of dimension equal to $r$ is given by the Gaussian binomial
coefficient ${\sqbinom{n}{r}}_2 = \prod_{i = 0}^{r-1} \frac{(2^n-2^i)}{(2^r-2^i)} \leq \frac{2^{nr}}{2^{r^2}}$. Thus, by a union bound, we get the following. 
\[
\Pr_{y}\left[\rank(y) = r \right] \leq \frac{2^{nr}}{2^{r^2}} \cdot 2^{-n(n-r)} = 2^{-(n-r)^2}
 \, .\qedhere\]
\end{proof}

\subsection*{Acknowledgements}
We would like to thank Suryateja Gavva for helpful discussions. 
We would like to thank Shubhangi Saraf for suggesting the idea 
for the proof of \cref{lem:bias ub for random tensor}.

{\small
\bibliographystyle{prahladhurl}
\bibliography{ref}

}
\appendix
\input{luckyinequality}

\end{document}

%% file: luckyinequality.tex
\section{A maximization problem}
In this section, we prove~\cref{thm:max}. We start with restating it here.
\begin{theorem}[Restatement of~\cref{thm:max}]
Let $k$ be a positive integer, and let $u \in [0, k^2]$ be a real number.
Suppose $b_1, b_2, \ldots, b_k$ are real numbers satisfying the following constraints.
\begin{align}
k \geq b_1 \geq b_2 \ldots \geq b_k \geq 0, \label{ordereq1}\\
\sum_{i=1}^k b_i = u. \label{sumeq1}
\end{align}
Then,
$$ \sum_{i=1}^{k} 2^{i-1} 2^{b_i} \leq \sum_{i=1}^{k} 2^{i-1} 2^{u/k} = (2^k-1) 2^{u/k}.$$
\end{theorem}
\begin{proof}
Let $\calP$ denote the convex polytope defined as follows.
$$ \calP = \{ (x_1, \ldots, x_n) \in \mathbb R^n \mid k \geq x_1 \geq \ldots \geq x_k \geq 0 \mbox{ and } \sum_i x_i = u \}.$$
Let $f : \mathbb R^n \to \mathbb R$ be the function:
$$ \sum_{i=1}^k 2^{i-1} 2^{x_i}.$$
Observe that $\calP$ is bounded and nonempty, and $f$ is a convex function.
Thus the maximum $M$ of $f$ on $\calP$ is achieved at an extreme point.
Since $\calP$ is defined by $k+1$ inequalities and $1$ equality,
extreme points satisfy the $1$ equality and make at least $k-1$ of the inequalties tight.
Thus any extreme point  $(y_1, y_2, \ldots, y_k)$ of $\calP$ satisfies, for some
integers $a, b,c \geq 0$ with $a+b+c = k$, and some $\ell \in (0,k)$, the following equalities.
$$ y_1 = y_2 = \ldots = y_a = k,$$
$$y_{a+1} = y_{a+2} = \ldots = y_{a+b} = \ell,$$
$$ y_{a+b+1} = y_{a+b+2} = \ldots = y_{k} = 0.$$
$$ ak + b\ell = u.$$
At such an extreme point $(y_1, \ldots, y_k)$, the value of $f$ can be expressed in terms of $a,b,c,\ell$ as
\begin{align*}
f(y_1, \ldots, y_k) &= 2^k (2^{a}-1) + 2^{\ell} 2^{a} (2^b - 1) + 2^0 2^{a+b} (2^c - 1)\\
&= 2^{a} (2^k - 2^\ell) + 2^{a + b} (2^\ell - 1).
\end{align*}
The following lemma then completes the proof of the theorem.

\begin{lemma}
Let $k$ be a positive integer, and let $u \in [0, k^2]$ be a real number.
Let $a, b, c, \ell \in [0, k]$ be real numbers with:
$$a+b+c =k,$$
$$ak + b\ell = u.$$
Then 
$$2^{a} (2^k - 2^\ell) + 2^{k-c} (2^\ell - 1) \leq (2^k - 1)2^{u/k}.$$
\end{lemma}
\begin{proof}
Let $\alpha, \beta, \gamma, \lambda, \eta \in [0,1]$ be given by
\begin{align*}
 a= \alpha k, b = \beta k,  c = \gamma k, \ell = \lambda k, u = \eta k^2.
\end{align*}
Then, we have
\begin{align}
\alpha + \beta + \gamma = 1,\label{eqcond1}\\
 \alpha + \beta \lambda = \eta. \label{eqcond2}
\end{align}
Let $Z = 2^k$.
Then we want to show that whenever $\alpha, \beta, \gamma, \lambda, \eta$ are as above,
we have
$$ Z^{\alpha} ( Z - Z^{\lambda}) + Z^{1-\gamma} ( Z^{\lambda} - 1) \leq (Z-1) Z^{\eta}.$$
Eliminating $\alpha, \gamma$ from Equation~\eqref{eqcond1} and Equation~\eqref{eqcond2},
we have $ \gamma = (1 - \eta) - \beta(1-\lambda)$, and $\alpha = \eta - \beta \lambda$.
Substituting this in, we want to show that
$$ Z^{\eta - \beta \lambda} (Z - Z^{\lambda}) + Z^{\eta + \beta(1-\lambda)} (Z^{\lambda} - 1) \leq (Z-1) Z^{\eta}.$$
Dividing throughout by $Z^\eta$, we want to show that
$$ Z^{-\beta \lambda} (Z - Z^{\lambda}) + Z^{\beta ( 1- \lambda)}(Z^{\lambda} - 1) \leq Z-1.$$
Rewriting, this is the same as
$$ Z - Z^{\lambda} + Z^{\beta + \lambda} - Z^{\beta} \leq (Z -1) Z^{\beta \lambda},$$
which is equivalent to 
$$ (Z^\beta - 1 ) (Z^\lambda - 1) \leq (Z^{\beta \lambda} - 1) (Z-1).$$
This follows from Lemma~\ref{lem:betalambdaineq}.
\end{proof}
This completes the proof.
\end{proof}

\section{Numerical Inequalities}

In this section we list some numerical inequalites that are used in the previous section.

\begin{lemma}
\label{lem:increasingineq}
For all real $r \geq 1$, the function $f : [1, \infty) \to \mathbb R$ given by
$$ f(x)= \frac{ x^r - 1}{x-1}$$
is increasing in $x$.
\end{lemma}
\begin{proof}
We show that $f'(x) \geq 0$ for all $x \geq 1$.
Compute
$$ f'(x) = \frac{  (r x^{r-1}) \cdot (x-1) - (x^r - 1) \cdot 1 }{(x-1)^2}$$
Define 
$$g(x) = r (x^{r} - x^{r-1}) - (x^r - 1) = (r-1) x^r - r x^{r-1} + 1.$$
The positivity of $f'(x)$ would follow if we can show:
$$ g(x) \geq 0$$
for all $x \geq 1$.
We prove this by first observing that $g(1) = 0$,
and then showing that for all $x \geq 1$, we have $g'(x) \geq 0$.
Indeed, 
\begin{align*}
g'(x) &= r (r-1) x^{r-1} - r (r-1) x^{r-2} \\
&= r(r-1) (x^{r-2}) (x - 1) \\
& \geq 0.
\end{align*}
This completes the proof that $g(x) \geq 0$ for all $x \geq 1$, and
thus the proof that $f'(x) \geq 0$ for all $x \geq 1$.
\end{proof}

\begin{lemma}
\label{lem:betalambdaineq}
For all real $z \geq 1$ and all real $\beta, \lambda \in [0, 1]$, we have:
$$ (z^{\lambda} - 1) (z^{\beta} - 1) \leq (z^{\beta \lambda} - 1) ( z- 1).$$
\end{lemma}
\begin{proof}
If either $\lambda = 0$ or $z = 1$, the inequality trivially holds (with equality).
Now suppose $\lambda \neq 0$ and $y \neq 1$.
Set $r = 1/\lambda$ and $x = z^{\lambda}$ and $y = z^{\beta \lambda}$.
Then, $1 \leq  y\leq x$ and $r \geq 1$.
Then, the inequality we want to prove can be written as follows.
$$ (x-1) (y^r -1) \leq (y-1) (x^r -1),$$
i.e.,
$$ \frac{y^r - 1}{y-1} \leq \frac{x^r - 1}{x-1}.$$
This follows from Lemma~\ref{lem:increasingineq}, completing the proof.
\end{proof}